\documentclass[prl,aps,twocolumn,a4paper]{revtex4-1}
\usepackage{amsmath,amsthm,amssymb,amsfonts,mathrsfs,amsthm}
\usepackage{dsfont,yfonts,calligra}
\usepackage{graphicx,hyperref}
\usepackage{pxfonts,upgreek}
\usepackage{lmodern}
\usepackage{enumitem}

\newtheorem*{theorem*}{Theorem}

\DeclareMathAlphabet{\mathpzc}{OT1}{pzc}{m}{it}
\DeclareMathAlphabet{\mathcalligra}{T1}{calligra}{m}{n}

\def\Tr{\operatorname{Tr}} \def\>{\rangle} \def\<{\langle}
 \def\id{\mathsf{id}}
 
 \def\mE{\mathcal{E}}

\def\mS{\mathcal{S}}
\def\sH{\mathcal{H}}  
\def\sS{{\boldsymbol{\mathsf{S}}}}  
\def\bound{\boldsymbol{\mathsf{L}}}

 \def\dec#1{\mathscr{D}}
 \def\openone{\mathds{1}}

\def\mP{\mathcal{P}}
\def\mW{\mathcal{W}}
\def\mR{\mathcal{R}}

\renewcommand{\qedsymbol}{\nobreak \ifvmode \relax \else
  \ifdim \lastskip<1.5em \hskip-\lastskip \hskip1.5em plus0em
  minus0.5em \fi \nobreak \vrule height0.75em width0.5em
  depth0.25em\fi}

\renewcommand{\ge}{\geqslant}
\renewcommand{\le}{\leqslant}

\newtheorem{theorem}{Theorem}

\theoremstyle{remark}

\theoremstyle{definition}

\newcommand{\bea}{\begin{eqnarray}}
\newcommand{\eea}{\end{eqnarray}}
\newcommand{\be}{\begin{equation}}
\newcommand{\ee}{\end{equation}}

\begin{document}



\title{On complete positivity, Markovianity, and the quantum data-processing inequality,\\ in the presence of initial system-environment correlations}

\author{Francesco Buscemi}

\affiliation{Graduate School of Information Science, Nagoya University, Chikusa-ku, Nagoya 464-8601, Japan}

	\email{buscemi@is.nagoya-u.ac.jp}






\begin{abstract}
	We show that complete positivity is not only sufficient but also necessary for the validity of the quantum data-processing inequality. As a consequence, the reduced dynamics of a quantum system are completely positive, even in the presence of initial correlations with its surrounding environment, if and only if such correlations do not allow any anomalous backward flow of information from the
  environment to the system. Our approach provides an intuitive information-theoretic framework to unify and extend a number of previous results.
\end{abstract}

\maketitle





In the case in which we are describing {global evolutions} as
``quantum processors'' or ``{input-output black boxes,}'' there is no doubt that the
only operationally, physically, and mathematically well-defined way to
proceed is that given by the formalism of {\textit{quantum operations}, in the sense of Kraus}~\cite{kraus,kraus2,nc}, i.e.,
completely positive (CP) linear
maps. As it turns out, quantum operations can always be
modeled as interactions of the input system with an environment,
initially factorized from (and independent of) the input system, and
discarded after the interaction took
place~\cite{stinespring,kraus,kraus2,ozawa,nc}. Such a model, however, is not
universally valid, but relies on an \emph{initial factorization
  condition}.

The question then naturally arises~\cite{pechukas, alicki, initial-corr, initial-corr2, initial-corr3,initial-corr4,initial-corr5,discord1,discord2,discord3,discord4}: what happens when the initial
factorization condition does not hold, namely, when system and
environment are, already before the interaction is turned on,
correlated? While this question arguably originated
from practical motivations (e.g., the difficulty to experimentally
enforce the initial factorization assumption), it soon moved to a more
fundamental level, in an attempt to challenge the very physical
arguments often put forth to promote CP dynamics as the only
``physically reasonable'' reduced dynamics. (On this point see, e.g.,
Refs.~\cite{kraus,kraus2}, but also Section 8.2.4 of~\cite{nc}). As one would
expect, by allowing the input system and its environment to start in a
correlated state, it is possible that the reduced dynamics of the
system are not CP anymore. The possibility of exploring phenomena
outside the CP framework attracted considerable
interest~\cite{initial-corr,initial-corr2,initial-corr3,initial-corr4,initial-corr5,discord1,discord2,discord3,discord4,thermo-viol,thermo-viol2,no-clon-viol,breuer,Rivas,Xu,Darrigo}, in particular in connection with the
possibility of circumventing thermodynamic or information-theoretic
\textit{tenet} like, e.g., the second law of thermodynamics (by
anomalous heat flow~\cite{thermo-viol,thermo-viol2}) or the data-processing
inequality (by anomalous increase of distinguishability~\cite{breuer}, {by entanglement revivals via local operations}~\cite{initial-corr2,Rivas,Xu,Darrigo},
or by violating the no-cloning theorem~\cite{no-clon-viol}). {In the language of the theory of open quantum systems, all such violations are interpreted as signatures of the fact that the underlying global evolution is \textit{non-divisible}~\cite{breuer,Rivas,Xu,Darrigo,Rivas-rev}, i.e., it cannot be decomposed into a chain of CP maps across successive time intervals. Fig.~\ref{fig:example} below illustrates a simple example of such a non-divisible evolution.}
\begin{figure}[bt]
	\includegraphics[width=8cm]{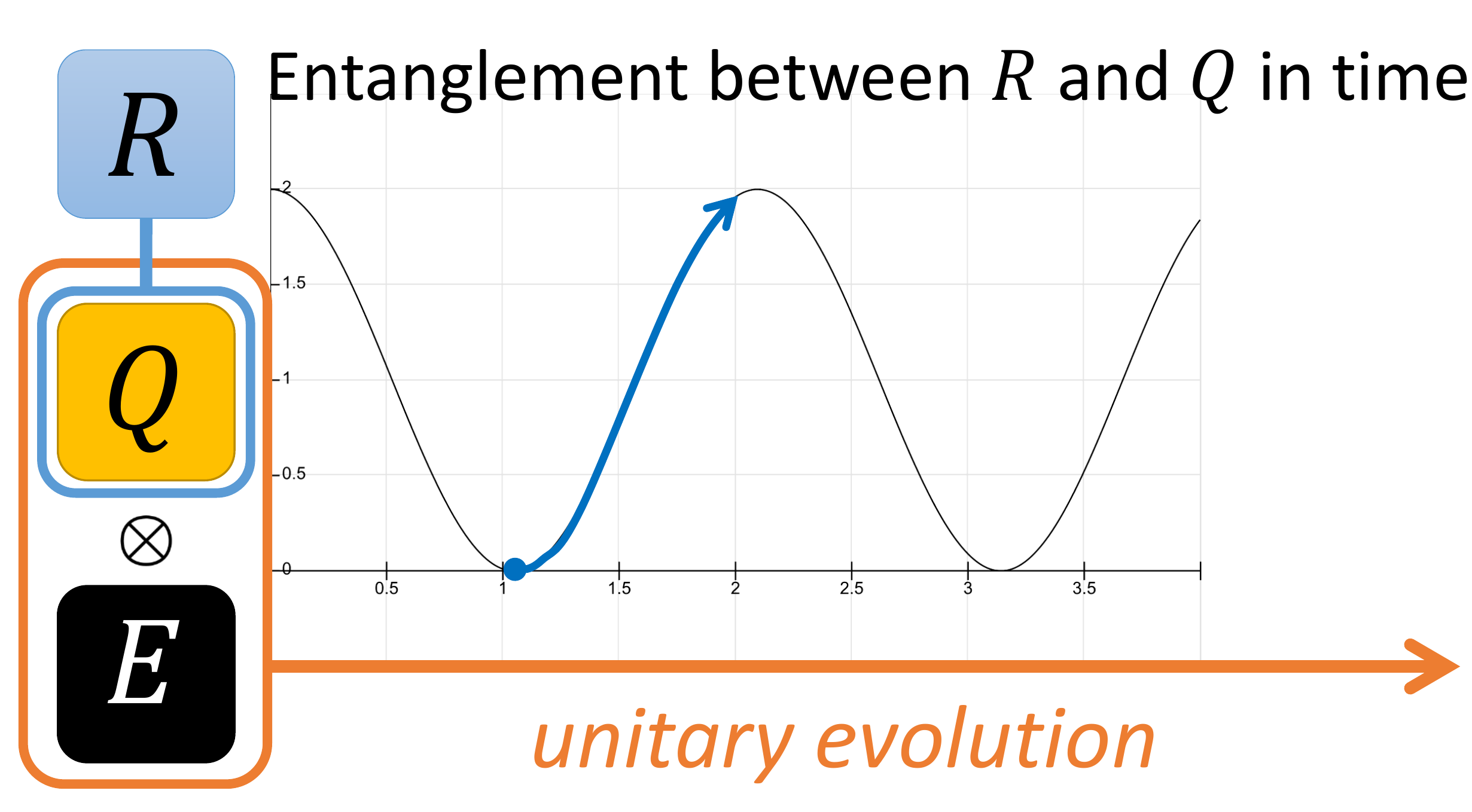}
	\caption{A simple example exhibiting non-CP reduced dynamics. The {quantum} system $Q$, initially entangled with a reference {quantum} system $R$, locally interacts with a {quantum} environment $E$, initially factorized from both $Q$ and $R$, while the reference remains isolated during the whole process. Nonetheless, as a consequence of the joint system-environment unitary evolution, the entanglement between $R$ and $Q$ may periodically oscillate between a maximum and a minimum value (the numerical values plotted are only indicative). An observer, having access to $R$ and $Q$ only, would notice periodic revivals of entanglement (the arrow in blue) between $R$ and $Q$ due to a \textit{local transformation} of $Q$ alone -- {a clear signature that the reduced dynamics of $Q$, within two instants in time during which such revivals occur, cannot be described by a CP map (i.e., locally in time, during the revivals, the evolution of $Q$ is \textit{non-divisible})}.}
	\label{fig:example}
\end{figure}
%

The study of non-CP dynamics naturally motivates also the
complementary search for conditions (to be satisfied by the
initial system-environment correlations) that guarantee CP reduced
dynamics. Within this trend we recall, in particular, the
Pechukas-Alicki debate about whether the initial factorization
condition is the only reasonable one to require (Ref.~\cite{alicki}
commenting on Ref.~\cite{pechukas}). More recently, attempts have been
made to show that CP reduced dynamics are fundamentally related with
initial system-environment correlations having vanishing quantum
discord: building on Ref.~\cite{discord1}, where it was proved that zero discord is sufficient for CP reduced dynamics, Ref.~\cite{discord2} later claimed the two conditions to be equivalent, but it turned out that the equivalence only holds for restricted spaces of initial states~\cite{discord3}. A rigorous framework for this problem has recently been provided in Ref.~\cite{discord4}, but
a complete characterization of which initial system-environment
correlations lead to CP reduced dynamics is not available yet.

Here, we propose a different way to tackle this problem starting from a simple initial idea:
if, on one hand, non-CP reduced dynamics lead to the violation of
data/energy-processing principles (as seen above), can complete
positivity, on the other hand, be recovered by \emph{assuming} that a
suitably chosen data-processing inequality \textit{always holds}? In other
words, is the absence of anomalous backward flows of information/heat
only necessary, or is it also sufficient for the complete positivity
of the reduced dynamics? Our answer, in the affirmative, not only
establishes a novel information-theoretic characterization of the
property of complete positivity, but provides, as a by-product, a
comprehensive description of which initial system-environment
correlations indeed lead to CP reduced dynamics.

Our approach, as in the example depicted in Fig.~\ref{fig:example}, does not focus solely on the initial system-environment
correlations, but brings into play a third system, a \emph{reference},
with respect to which both the property of complete positivity~\cite{choi} and the
data-processing inequality~\cite{schu-niel} can be conveniently phrased. Our results
are therefore formulated in terms of \emph{tripartite}, rather than bipartite, initial
  configurations: it is however a straightforward matter to go back,
from our tripartite framework, to the conventional bipartite
system-environment scenario.

\begin{figure}[b]
	\includegraphics[width=\columnwidth]{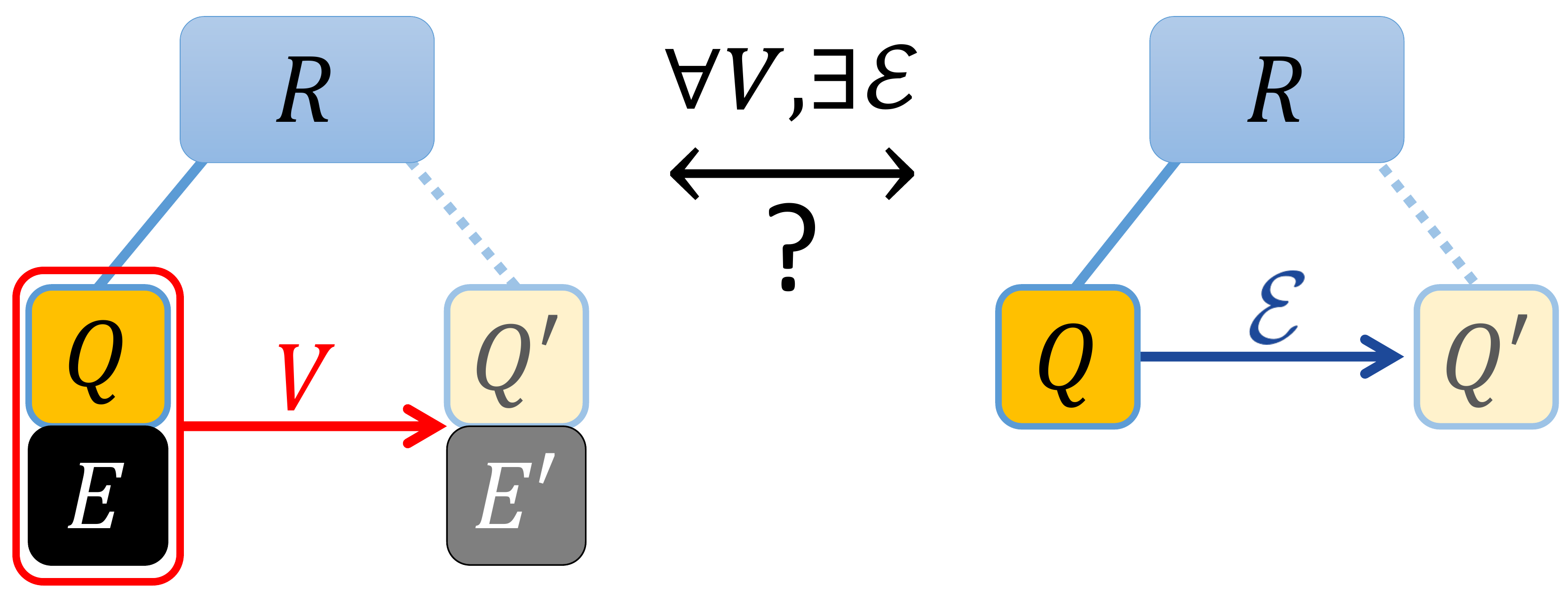}
	\caption{Our operational framework: given an initial tripartite configuration $\rho_{RQE}$, we ask whether, for any joint system-environment evolution $V:QE\to Q'E'$, there exists a completely positive, trace-preserving linear map $\mE$ describing the reduced dynamics $Q\to Q'$.}
	\label{fig:tripartite}
\end{figure}

To be more precise, the operational framework we adopt is the
following (see also Fig.~\ref{fig:tripartite} below):
\begin{enumerate}

\item {At some time $t=\tau$, we fix a tripartite configuration, i.e., an arbitrary tripartite
	density operator $\rho_{RQE}^\tau$, describing the initial correlations
	between the system $Q$, its environment $E$, and a reference $R$. The reference, reminiscent of the construction used
	by Choi~\cite{choi} to study completely positive maps, here plays the role
	of the ``blind and dead'' \emph{witness system} of
	Pechukas~\cite{pechukas}. It can be helpful to imagine that we are in a situation like that in Fig.~\ref{fig:example}, but ``freezed'' at some intermediate time $t=\tau$, when the correlations between $R$, $Q$, and $E$ are arbitrary.}

\item We move to the next instant in time, i.e., $t=\tau+\Delta$. We assume that the pair system-environment evolves from $\tau$ to $\tau+\Delta$ according to some unitary operator $V$, while the reference $R$ remains unchanged.
  
\item {Denoting by $Q'$ and $E'$ the system and the environment after the
	evolution described by $V$ has taken place, the tripartite configuration $\rho_{RQE}^\tau$
	has evolved to the tripartite configuration
	$\sigma_{RQ'E'}^{\tau+\Delta}=(\openone_R\otimes
	V_{QE})\rho_{RQE}^\tau(\openone_R\otimes V_{QE}^\dag)$.}

\item {We then look at the reduced
	reference-system dynamics (i.e., the transformation mapping $\rho_{RQ}^\tau$ to $\sigma_{RQ'}^{\tau+\Delta}$) and check whether they are compatible with the application of a completely positive trace-preserving linear map on
	the system $Q$ \emph{alone}. More explicitly, we check whether there exists a completely positive trace-preserving linear map
	$\mE$, mapping $Q$ to $Q'$, such that
	$\sigma_{RQ'}^{\tau+\Delta}=(\id_R\otimes\mE_Q)(\rho_{RQ}^\tau)$.}

\end{enumerate}

The main result of this paper is to prove that those
tripartite configurations, for which the reduced reference-system
dynamics are always (i.e., for all possible evolutions $V$)
compatible with the application of a completely positive trace-preserving map on the system $Q$ alone,
are exactly those tripartite configurations, for which the
reduced reference-system dynamics never violates the \textit{quantum data-processing
	inequality}~\cite{schu-niel}. This result is obtained by proving
that such tripartite configurations are exactly those constituting
\emph{short quantum Markov chains}~\cite{hayden-petz}, i.e., tripartite states such that
the quantum mutual information between the reference and the
environment, conditional on the system, is
zero. The rest of the paper is devoted to carefully
define all the above ideas, state the main results, and show how our tripartite scenario provides a way to unify and, at the same time, considerably
extend a number of results previously considered in the literature.

\medskip


\textit{Useful facts.}---In what follows, we only consider quantum
systems defined on finite dimensional Hilbert spaces $\sH$. We denote
by $\bound(\sH)$ the set of all linear operators acting on $\sH$, by
$\bound^+(\sH)\subset\bound(\sH)$ the set of all positive
semi-definite elements, and by $\sS(\sH)\subset \bound^+(\sH)$ the set
of all states $\rho$, i.e., operators with $\rho\ge 0$ and
$\Tr[\rho]=1$. The identity operator in $\bound(\sH)$ will be denoted
by the symbol $\openone$, whereas the identity map from $\bound(\sH)$
to $\bound(\sH)$ will be denoted by $\id$. In what follows, a
\emph{channel} is meant to be a linear map $\mE:\bound(\sH_Q)\to
\bound(\sH_{Q'})$, which is completely positive and trace-preserving
(CPTP) everywhere. For brevity, we will denote $\mE:\bound(\sH_Q)\to
\bound(\sH_{Q'})$ simply by $\mE:Q\to Q'$. Also, a \emph{joint evolution}
of systems $Q$ and $E$ will be meant to be an isometry
$V:\sH_Q\otimes\sH_E\to \sH_{Q'}\otimes\sH_{E'}$, or, for brevity,
$V:QE\to Q'E'$.


As anticipated above, we consider physical situations, in which the
global configuration can be divided into a reference $R$, the system
$Q$, and its surrounding environment $E$. A particular configuration, at some instant in time $\tau$ that we choose as the initial time for our analysis,
is then specified by assigning a tripartite state
$\rho_{RQE}\in\sS(\sH_R\otimes\sH_Q\otimes\sH_E)$. Among all
tripartite states, those that will play a central role here are the
so-called short quantum Markov chains (\emph{Markov states}, for
short), namely those states for which the conditional quantum mutual
information, defined as
\begin{equation}\label{eq:cqmi}
I(R;E|Q)_\rho=S(\rho_{RQ})+S(\rho_{QE})-S(\rho_{RQE})-S(\rho_Q),
\end{equation}
$S(\rho):=-\Tr[\rho\ \log\rho]$ being the usual von Neumann entropy,
is null. The key property of Markov states we need is the following~\cite{hayden-petz}: a tripartite state
$\rho_{RQE}$ is such that $I(R;E|Q)_\rho=0$ if and only if there exists a CPTP map $\mR:Q\to QE$ such that $\rho_{RQE}=(\id_R\otimes\mR_Q)(\rho_{RQ})$.
\medskip

\textit{The data-processing inequality.}---Consider three classical
random variables, $X$, $Y$, and $Z$, constituting a Markov chain $X\to Y\to Z$. Then, by defining the input-output mutual informations
$I(X;Y)=H(X)+H(Y)-H(XY)$ and, analogously, $I(X;Z)$, the data-processing inequality $I(X;Y)\ge
I(X;Z)$ holds~\cite{cover}. This formalizes the
intuition that any post-processing can only decrease the total amount
of information, or, in other words, that ``there are no free lunches
in information theory.''

The same intuition is indeed correct also in quantum information
theory and is formalized as follows.
One first introduces a reference system $R$, initially correlated
with the input quantum system $Q$, their joint state being denoted by
$\rho_{RQ}$. Then, the system $Q$ is fed into the channel $\mE:Q\to Q'$, while the reference is left
untouched. The joint state $\rho_{RQ}$ has correspondingly changed
into $\sigma_{RQ'}:=(\id_R\otimes\mE_Q)(\rho_{RQ})$. The quantum data-processing theorem~\cite{schu-niel} states that, for any $\rho_{RQ}\in\sS(\sH_R\otimes\sH_Q)$ and for any CPTP map
$\mE:Q\to Q'$, the following inequality holds:
\begin{equation}\label{eq:qdp}
	I(R;Q)_\rho\ge I(R;Q')_\sigma.
\end{equation}
where $I(R;Q)_\rho:=S(\rho_R)+S(\rho_Q)-S(\rho_{RQ})$ and
$I(R;Q')_\sigma:=S(\sigma_R)+S(\sigma_{Q'})-S(\sigma_{RQ'})$ are, respectively, the initial and final quantum mutual informations.

In other words, the quantum data-processing theorem states that any
local post-processing can only decrease the total amount of
correlations. This constitutes the formal reason why any revival of correlations, like that described in Fig.~\ref{fig:example}, falls outside the framework of CP dynamics.

The reference system $R$ plays, in the quantum
data-processing inequality, the same role that the input random
variable plays in the classical data-processing inequality. The
relation between these two views, seemingly rather different, is given
by the possibility of using the system $R$ to ``steer''~\cite{steering,steering2} different
states on $Q$ by means of local measurements on $R$ alone. The
set of states $\mS$ that can be steered on $Q$ from $R$ depends on the
initial correlations between $R$ and $Q$, i.e., it depends on the joint
state $\rho_{RQ}$, as follows:
\begin{equation}\label{eq:steering}
  \mS_Q(\rho_{RQ}):=\left\{\frac{\Tr_R[(P_R\otimes\openone_Q)\
  \rho_{RQ}]}{\Tr[(P_R\otimes\openone_Q)\
  \rho_{RQ}]}:P_R\in\bound^+(\sH_R)\right\}.
\end{equation}
We will use steering in order to recover, from the
tripartite scenario (reference-system-environment) employed in this
paper, the simpler and more conventional bipartite scenario (system-environment) usually
considered in the literature.\medskip

\textit{Assuming the data-processing inequality.}---Let the starting configuration be given by a tripartite state
$\rho_{RQE}$ shared between a reference $R$, the quantum system under
observation $Q$, and the environment $E$. Notice that we do not make
any assumptions on the initial correlations existing among $R$, $Q$,
and $E$. We then let the system $Q$ and the environment $E$ evolve jointly: this is
formalized by applying an isometry
$V:QE\to Q'E'$. We then trace out the
final environment system $E'$, focusing on the reduced dynamics $RQ\to
RQ'$ (see again Fig.~\ref{fig:tripartite}).

We begin with a definition: given an initial tripartite configuration $\rho_{RQE}$ and a joint
  system-environment evolution $V:QE\to Q'E'$, define the final state
  $\sigma_{RQ'E'}:= V_{QE}\rho_{RQE}V_{QE}^\dag$. We say that the
  reduced dynamics $Q\to Q'$ are {(globally)} CPTP if there exists a CPTP map
  $\mE:Q\to Q'$ such that
\begin{equation}\label{eq:red-CPTP}  
\sigma_{RQ'}=(\id_R\otimes\mE_Q)(\rho_{RQ}).
\end{equation}
We are now ready to state the main theorem (whose proof can be found in the Supplemental Material~\cite{suppl}):
\begin{theorem}\label{theo:main}
  Fix a tripartite configuration $\rho_{RQE}$. The following are
  equivalent:
\begin{enumerate}[label=\emph{(\alph*)}]
\item For any joint system-environment evolution $V:QE\to Q'E'$, the
  reduced dynamics $Q\to Q'$ satisfy the quantum data-processing
  inequality~(\ref{eq:qdp}).
\item $I(R;E|Q)_\rho=0$.
\item For any joint system-environment evolution $V:QE\to Q'E'$, the
  reduced dynamics $Q\to Q'$ are CPTP, in the sense
  of~(\ref{eq:red-CPTP}).
\end{enumerate}
\end{theorem}

The rest of the paper is devoted to showing how considerable parts of
the previous literature can be seen as special cases of the tripartite
scenario considered here. The central idea is that, considering a
tripartite configuration $\rho_{RQE}$ is essentially equivalent to
considering a \textit{whole family} of bipartite states on
$\sH_Q\otimes\sH_E$. Such a family is obtained by ``steering'' states
of $QE$ by measurements on $R$, as formally described in
Eq.~(\ref{eq:steering}). We start with a definition: we say that a
family of bipartite system-environment states
$\mS\subseteq\sS(\sH_Q\otimes\sH_E)$ is a \textit{post-selected} (or
  \textit{post-selectable}) family, if there exist an initial system $Q_0$, an
initial environment $E_0$, and a completely positive map
$\mP:Q_0E_0\to QE$ such that $\mS$ coincides with the output state
space of $\mP$, i.e.,
	\begin{equation}\label{eq:post-sele}
	\mS=\left\{\frac{\mP(\rho_{Q_0E_0})}{\Tr[\mP(\rho_{Q_0E_0})]}:\rho_{Q_0E_0}\in\sS(\sH_{Q_0}\otimes\sH_{E_0})\right\}.
	\end{equation}
It is important to stress that the class of post-selected families is very general: it contains, for example, all families that are obtained as mixtures of a finite or countable number of fixed states (i.e., polytopes), but also families of the form $\rho_Q\otimes\bar\sigma_E$, for varying $\rho$ and fixed $\bar\sigma$. It is also important to stress that the idea of post-selection introduced here is rather different from that considered by Alicki in~\cite{alicki}: while Alicki insists on post-selection operations embedding $Q$ into $QE$, here we lift such limiting assumption, allowing instead for a completely general initial system $Q_0E_0$.

We are now ready to state the second main result of this paper (whose proof, uses some ideas previously employed in Refs.~\cite{buscemi1,buscemi2}, can be found in the Supplemental Material~\cite{suppl}):

\begin{theorem}\label{theo:2}
  Let $\mS\subseteq\sS(\sH_Q\otimes\sH_E)$ be a post-selected family of initial
  bipartite system-environment states, possibly correlated, as in~(\ref{eq:post-sele}). The following are equivalent:
  \begin{enumerate}[label=\emph{(\alph*)}]
  	\item There exists a reference system $\sH_R$ and a tripartite Markov state
  	$\rho_{RQE}$ such that
  	$\mS_{QE}(\rho_{RQE})=\mS$.
  	\item For any joint system-environment evolution $V:QE\to Q'E'$, there exists a corresponding CPTP map $\mE:Q\to Q'$
  	such that
  	$\Tr_{E'}[V_{QE} \omega_{QE}V_{QE}^\dag]=\mE(\Tr_E[\omega_{QE}])$, for all
  	$\omega_{QE}\in \mS$.
  \end{enumerate}
\end{theorem}
Theorem~\ref{theo:2} above states that, checking whether a given family of bipartite system-environment states gives rise to completely positive reduced dynamics, is equivalent to checking whether said family can be obtained by steering from a tripartite Markov state. While we cannot say that the latter condition is always easier to check than the former, it is fair to affirm that the latter \textit{is} very easy to check for the relevant examples presented so far in the literature.

\medskip\textit{First example: families with zero quantum discord.}---Let us consider bipartite system-environment
states of the form:
\begin{equation}\label{eq:c-q}
\rho_{QE}^p=\sum_{i=1}^n p_i |i_Q\>\<i_Q|\otimes\rho^i_E,
\end{equation}
where $\{|i_Q\>\<i_Q|\otimes\rho^i_E\}_{i=1}^n$ are $n$ fixed bipartite states such that $\<i_Q|j_Q\>=\delta_{ij}$, while
$p=(p_1,\cdots,p_n)$ is an arbitrary vector of probabilities. For varying $p$, Eq.~(\ref{eq:c-q}) defines a family of states, which we denote by $\mS_0$. Notice that the family $\mS_0$ is a finite polytope and, therefore, Theorem~\ref{theo:2} can be applied. The family $\mS_0$ was first considered, in relation with the problem of characterizing complete positivity, in Ref.~\cite{discord1}, where it is proved, by direct inspection of the coefficient matrix of the associated dynamical maps, that $\mS_0$ satisfies condition (b) of Theorem~\ref{theo:2}. In this case, however, condition (a) seems much easier to check, since a tripartite Markov state inducing $\mS_0$ by steering on $R$ is simply given by, e.g.,
\begin{equation*}
\rho_{RQE}=\frac 1n\sum_i
|i_R\>\<i_R|\otimes|i_Q\>\<i_Q|\otimes\rho^i_E,
\end{equation*}
for $\<i_R|j_R\>=\delta_{ij}$.\medskip

\textit{Second example: families with non-zero quantum discord}.---Due to the fact that states of the
form~(\ref{eq:c-q}) are also those (and only those) with zero quantum
discord~\cite{hayashi}, the absence of
quantum discord has been conjectured to be the essential ingredient for complete
positivity. Ref.~\cite{discord2} later claimed to have a proof of the conjecture, but that claim has recently been disproved in Ref.~\cite{discord3}. The counterexample provided in Ref.~\cite{discord3} is constructed from
bipartite system-environment states of the following form:
\begin{equation}\label{eq:counter}
\zeta_{QE}^p=p\zeta^{\alpha}_{QE}+(1-p)\zeta^{\beta}_{QE},
\end{equation}
where $\zeta^{\alpha}_{QE},\zeta^{\beta}_{QE}\in\sS(\sH_Q\otimes\sH_E)$ are chosen as
\[
\zeta^{\alpha}_{QE}=\frac12|0_Q\>\<0_Q|\otimes\rho^0_E+\frac12|+_Q\>\<+_Q|\otimes\rho^+_E,
\]
with $|+\>=2^{-1/2}(|0\>+|1\>)$ and $\rho^0_E\neq\rho^+_E$, and
\[
\zeta^{\beta}_{QE}=|2_Q\>\<2_Q|\otimes\rho^2_E.
\]
The polytope of states obtained from~(\ref{eq:counter})
for $p$ varying in $[0,1]$ indeed satisfies condition (b) of Theorem~\ref{theo:2}, even though its elements can have non-zero quantum discord~\cite{discord3}. Ref.~\cite{discord3}, in order to prove (b), employs the same technique used in Ref.~\cite{discord1}, i.e., a direct inspection of the coefficient matrix of the associated dynamical maps. Also in this case, condition (a) seems much easier to check, simply by considering a tripartite Markov state of the form, e.g.,
\begin{equation}\label{eq:zeta}
\zeta_{RQE}=\frac12|\alpha_R\>\<\alpha_R|\otimes\zeta^\alpha_{QE}+\frac12|\beta_R\>\<\beta_R|\otimes\zeta^\beta_{QE},
\end{equation}
with $\<\alpha_R|\beta_R\>=0$.
For the sake of brevity, we moved the proof that $\zeta_{RQE}$ indeed is a Markov state to the Supplemental Material~\cite{suppl}. On the other hand, it is trivial to see that any state in~(\ref{eq:counter}) can be obtained from $\zeta_{RQE}$ by steering on $R$.\medskip

\textit{Third example: families of entangled states}.---The families
of states considered in Refs.~\cite{discord1} and~\cite{discord2}
contain only separable states. One may wonder whether separability is
necessary for completely positive reduced dynamics. In the
Supplemental Material~\cite{suppl} we provide an explicit procedure to construct
examples of families of system-environment that, though entangled,
still satisfy condition (b) of Theorem~\ref{theo:2}.\medskip

\textit{Fourth example: derivation of the initial factorization
  condition.}---We conclude by showing how the initial factorization
condition, usually \textit{assumed} in quantum information theory, can
in fact be \textit{derived} from the quantum data-processing
inequality and one natural extra assumption, namely, that the
restriction to the system $Q$ of the initial family
$\mS\subseteq\sS(\sH_Q\otimes\sH_E)$ of system-environment states
coincides with the \textit{whole} system's state space
$\sS(\sH_Q)$. This assumption in fact implies that the tripartite
state $\rho_{RQE}$ has the form
$|\Phi^+_{R'Q}\>\<\Phi^+_{R'Q}|\otimes\tilde{\rho}_{R''E}$, where
$R\cong R'R''$ and $|\Phi^+_{R'Q}\>$ is a pure entangled state with
invertible reduced density matrix. But then, the Markov condition
$I(R;E|Q)_\rho=0$ can hold only if $R''$ is the trivial (i.e.,
one-dimensional) system, i.e.,
$\rho_{RQE}=|\Phi^+_{R'Q}\>\<\Phi^+_{R'Q}|\otimes\tilde{\rho}_{E}$,
thus recovering the initial factorization condition.

\medskip\textit{Acknowledgements.}---The author is grateful to Michele Dall'Arno, Pawe{\l} Horodecki, Masanao Ozawa, David
Reeb, and Mark M. Wilde for insightful discussions on the topic of
this paper.

\appendix


\onecolumngrid
\newpage

\setcounter{page}{1}
\renewcommand{\thepage}{Supplemental Material -- \arabic{page}/3}
\setcounter{equation}{0}
\renewcommand{\theequation}{S.\arabic{equation}}

We recall, for completeness, the conventions we adopt. In what follows, we only consider quantum
systems defined on finite dimensional Hilbert spaces $\sH$. We denote
by $\bound(\sH)$ the set of all linear operators acting on $\sH$, by
$\bound^+(\sH)\subset\bound(\sH)$ the set of all positive
semi-definite elements, and by $\sS(\sH)\subset \bound^+(\sH)$ the set
of all states $\rho$, i.e., operators with $\rho\ge 0$ and
$\Tr[\rho]=1$. The identity operator in $\bound(\sH)$ will be denoted
by the symbol $\openone$, whereas the identity map from $\bound(\sH)$
to $\bound(\sH)$ will be denoted by $\id$. In what follows, a
\emph{channel} is meant to be a linear map $\mE:\bound(\sH_Q)\to
\bound(\sH_{Q'})$, which is completely positive and trace-preserving
(CPTP) everywhere. For brevity, we will denote $\mE:\bound(\sH_Q)\to
\bound(\sH_{Q'})$ simply by $\mE:Q\to Q'$. Also, a \emph{joint evolution}
between systems $Q$ and $E$ will be meant to be an isometry
$V:\sH_Q\otimes\sH_E\to \sH_{Q'}\otimes\sH_{E'}$, or, for brevity,
$V:QE\to Q'E'$.

\section{Proof of Theorem~\ref{theo:main}}

\begin{theorem*}
	Fix a tripartite configuration $\rho_{RQE}$. The following are
	equivalent:
	\begin{enumerate}[label=\emph{(\alph*)}]
		\item For any joint system-environment evolution $V:QE\to Q'E'$, the
		reduced dynamics $Q\to Q'$ satisfy the quantum data-processing
		inequality~(\ref{eq:qdp}).
		\item $I(R;E|Q)_\rho=0$.
		\item For any joint system-environment evolution $V:QE\to Q'E'$, the
		reduced dynamics $Q\to Q'$ are CPTP, in the sense
		of~(\ref{eq:red-CPTP}).
	\end{enumerate}
\end{theorem*}

\begin{proof}
	Let us start by proving that (a) implies (b). We consider the
quantum data-processing inequality in the form~(\ref{eq:qdp}), and we assume by hypothesis
that it holds for any joint evolution $V:QE\to Q'E'$. In particular, one
example of an evolution is one for which $Q' = QE$ (that is, the
isometry that simply embeds the joint system $QE$ of the input into the
output system $Q'$, while taking $E'$ to be a trivial system). Since
we have assumed that the data-processing inequality holds, it
follows that $I(R;Q)_\rho \ge I(R;Q')_\sigma = I(R;QE)_\rho$. On the
other hand, by tracing over system $E$ of state $\rho_{RQE}$, the
quantum data-processing inequality also gives $I(R;QE)_\rho \ge
I(R;Q)_\rho$. Thus, it follows from our assumption that
$I(R;QE)_\rho = I(R;Q)_\rho$, which is equivalent to $I(R;E|Q)_\rho
= 0$.

Let us now prove that (b) implies (c), that is, we assume that
$I(R;E|Q)_\rho=0$ and show that, for any isometry
$V:QE\to Q'E'$, there exists a CPTP map $\mE:Q\to Q'$ such that
$\sigma_{RQ'}=(\id_R\otimes\mE_Q)(\rho_{RQ})$. First of all, we
recall that the condition $I(R;E|Q)_\rho=0$ guarantees the existence of a 
CPTP map $\mR:Q\to QE$ such that $\rho_{RQE}=(\id_R\otimes\mR_Q)(\rho_{RQ})$~\cite{hayden-petz}. We therefore start from the reduced reference-system state $\rho_{RQ}=\Tr_E[\rho_{RQE}]$, on which we apply (locally on $Q$) the reconstruction map $\mR$ to obtain $\rho_{RQE}$. We then apply the
same isometry $V:QE\to Q'E'$, finally tracing over
$E'$, leaving only $Q'$. The final result, by construction,
equals exactly that one would get from the original
system-environment evolution, but it has been obtained only by
local, CPTP transformations acting on system $Q$ alone. We therefore
explicitly constructed a CPTP map $\mE:Q\to Q'$ such that
$\sigma_{RQ'}=(\id_R\otimes\mE_Q)(\rho_{RQ})$.

The final logical relation, (c) implies (a), is nothing but the
quantum data-processing theorem~\cite{schu-niel}.
\end{proof}

\section{Proof of Theorem~\ref{theo:2}}

\begin{theorem*}
	Let $\mS\subseteq\sS(\sH_Q\otimes\sH_E)$ be a post-selected family of initial
	bipartite system-environment states, possibly correlated, as in~(\ref{eq:post-sele}). The following are equivalent:
	\begin{enumerate}[label=\emph{(\alph*)}]
		\item There exists a reference system $\sH_R$ and a tripartite Markov state
		$\rho_{RQE}$ such that
		$\mS_{QE}(\rho_{RQE})=\mS$.
		\item For any joint system-environment evolution $V:QE\to Q'E'$, there exists a corresponding CPTP map $\mE:Q\to Q'$
		such that
		$\Tr_{E'}[V_{QE} \omega_{QE}V_{QE}^\dag]=\mE(\Tr_E[\omega_{QE}])$, for all
		$\omega_{QE}\in \mS$.
	\end{enumerate}
\end{theorem*}

\begin{proof}
	Let us start by proving that (a) implies (b). By Theorem~\ref{theo:main}, we know that, for any joint system-environment evolution $V:QE\to Q'E'$, there exists a CPTP map $\mE:Q\to Q'$ such that \[\Tr_{E'}[V_{QE}\;\rho_{RQE}\;V^\dag_{QE}]=(\id_R\otimes\mE_Q)(\rho_{RQ}).\] This condition implies, in particular, that \[\Tr_{RE'}[(P_R\otimes\openone_{QE})\ V_{QE}\rho_{RQE}V^\dag_{QE}]= \Tr_R[(P_R\otimes\openone_{Q})\ (\id_R\otimes\mE_Q)(\rho_{RQ})],\] for any $P_R\in\bound^+(\sH_R)$. Since, on the other hand, we assume that the family $\mS$ can be obtained from $\rho_{RQE}$ by steering, we indeed obtain (b), i.e., for any $\omega_{QE}\in\mS$, $\Tr_{E'}[V_{QE}\omega_{QE}V_{QE}^\dag]=\mE(\Tr_E[\omega_{QE}])$.

Let us now turn to the converse, i.e., (b) implies (a). Let us introduce a reference system $R$, such that $\sH_R\cong\sH_{Q_0}\otimes\sH_{E_0}$, and define a maximally entangled state $|\Phi^+_{RQ_0E_0}\>\in\sH_R\otimes\sH_{Q_0}\otimes\sH_{E_0}$. It is a well-known fact that the maximally entangled state can be used to steer any state, i.e., $\mS_{Q_0E_0}(\Phi^+_{RQ_0E_0})=\sS(\sH_{Q_0}\otimes\sH_{E_0})$. Let us now construct the state \[\rho_{RQE}:=k(\id_R\otimes\mP_{Q_0E_0})(\Phi^+_{RQ_0E_0}),\] where $\mP:Q_0E_0\to QE$ is the completely positive map post-selecting the family $\mS$, and $k>0$ is a normalizing constant. Then, by construction, $\mS_{QE}(\rho_{RQE})=\mS$.

Let now \[\{P_R^i:1\le i\le (\dim\sH_R)^2\}\] be a complete set of linearly independent elements of $\bound^+(\sH_R)$. There exists therefore a corresponding family of self-adjoint operators \[\{L_R^i:1\le i\le (\dim\sH_R)^2\}\subset\bound(\sH_R)\] such that \[X_R=\sum_i\Tr[X_RP^i_R]L_R^i,\] for any $X_R\in\bound(\sH_R)$. By linearity then, \[\rho_{RQE}=\sum_iL^i_R\otimes\Tr_R[(P^i_R\otimes\openone_{QE})\ \rho_{RQE}]\] (see also Refs.~\cite{buscemi1,buscemi2}). In particular, for any $V:QE\to Q'E'$, \[V_{QE}\rho_{RQE}V_{QE}^\dag=\sum_iL^i_R\otimes\{V_{QE}\Tr_R[(P^i_R\otimes\openone_{QE})\ \rho_{RQE}]V_{QE}^\dag\}.\] We then use the fact that $\mS=\mS_{QE}(\rho_{RQE})$ to guarantee the existence of a CPTP map $\mE:Q\to Q'$ such that \[\Tr_E\{V_{QE}\ \Tr_R[(P^i_R\otimes\openone_{QE})\rho_{RQE}]\ V_{QE}^\dag\}=\mE_Q\{\Tr_{R}[(P^i_R\otimes\openone_{Q})\ \rho_{RQ}]\},\] for all $i$. Therefore, \[\Tr_E[V_{QE}\rho_{RQE}V_{QE}^\dag]=\sum_iL^i_R\otimes\mE_Q\{\Tr_{R}[(P^i_R\otimes\openone_{Q})\ \rho_{RQ}]\}=(\id_R\otimes\mE_Q)(\rho_{RQ}).\] This last condition is nothing but condition (c) in Theorem~\ref{theo:main}, which is equivalent to $I(R;E|Q)_\rho=0$.
\end{proof}

\section{Proof that the tripartite state in~(\ref{eq:zeta}) is Markovian}

We prove here that the state
\begin{equation*}
\zeta_{RQE}=\frac
12\bigg\{|\alpha_R\>\<\alpha_R|\otimes\bigg(\frac{|0_Q\>\<0_Q|\otimes\rho^0_E+|+_Q\>\<+_Q|\otimes\rho^+_E}{2}\bigg)\bigg\}+\frac 12\bigg\{|\beta_R\>\<\beta_R|\otimes |2_Q\>\<2_Q|\otimes\rho^2_E\bigg\},
\end{equation*}
with the condition $\<\alpha_R|\beta_R\>=0$ is Markovian, i.e.,
$I(R;E|Q)_\zeta=0$. To this end, consider two orthogonal projectors on
$Q$ defined as $P^\alpha_Q:=|0_Q\>\<0_Q|+|1_Q\>\<1_Q|$ and
$P^\beta_Q:=|2_Q\>\<2_Q|$. Clearly,
\begin{equation*}
\zeta_{RQE}=\sum_{x=\alpha,\beta}(\openone_R\otimes P^x_Q\otimes\openone_E)\
\zeta_{RQE}\ (\openone_R\otimes P^x_Q\otimes\openone_E).
\end{equation*}
Moreover
\begin{equation*}
\Tr_Q\left[(\openone_R\otimes P^x_Q\otimes\openone_E)\
\zeta_{RQE}\right]=\frac 12|x_R\>\<x_R|\otimes\zeta^x_E,\qquad x=\alpha,\beta,
\end{equation*}
where $\zeta^x_E=\Tr_Q[\zeta^x_{QE}]$. Therefore, conditional on the outcome $x$, the
reference $R$ and the environment $E$ are always factorized. This is
exactly the condition for which $I(R;E|Q)_\rho=0$.

\section{General method for constructing families of correlated system-environment states that always lead to complete positive reduced dynamics}

We present here a general procedure to explicitly obtain families
 of initial system-environment states which, even if highly
 entangled, nonetheless always lead to CP reduced dynamics. Such examples can be constructed as
 follows:
 \begin{enumerate}
 	
 	\item Fix a reference Hilbert space $\sH_R$, an initial Hilbert space $\sH_{Q_0}$, and a bipartite (possibly mixed) state $\omega_{RQ_0}\in\sS(\sH_R\otimes\sH_{Q_0})$.
 	
 	\item Choose now an \textit{invertible} CPTP map $\mW:Q_0\to Q$, i.e., such that there exists another CPTP map
 	$\mW':Q\to Q_0$ with $\mW'\circ\mW=\id_{Q_0}$.  Such CPTP maps, as shown
 	in~\cite{hayden-petz}, are such that their Stinespring isometric
 	dilations $W:\sH_{Q_0}\to \sH_{Q}\otimes\sH_{E}$, when locally applied to any
 	bipartite state $\omega_{RQ_0}$, automatically produce tripartite
 	states $\rho_{RQE}:=(\openone_R\otimes
 	W_{Q_0})\omega_{RQ_0}(\openone_R\otimes W^\dag_{Q_0})$ with
 	$I(R;E|Q)_\rho=0$.
 
 	\item From the tripartite Markov state $\rho_{RQE}$ obtained above, construct the family $\mS_{QE}(\rho_{RQE})\subseteq\sS(\sH_Q\otimes\sH_E)$, by steering from $R$, as described in Eq.~\ref{eq:steering}.
 	
 \end{enumerate}
 	
 Notice that, while the condition $I(R;E|Q)_\rho=0$ implies that $\rho_{RE}=\Tr_Q[\rho_{RQE}]$ is separable~\cite{hayden-petz}, no particular restriction exists for $\rho_{QE}$. This means that it is possible that the family $\mS_{QE}(\rho_{RQE})$ contains also entangled states. Nonetheless, by Theorem~\ref{theo:2}, any evolution $V:QE\to Q'E'$, followed by the partial trace over $E'$, will always lead to CPTP reduced dynamics $Q\to Q'$, for any state in $\mS_{QE}(\rho_{RQE})$.





\end{document}